\documentclass[preprint,12pt]{elsarticle}

\journal{Theoretical Computer Science}

\usepackage{amsmath,amssymb,amsfonts}
\usepackage{amsthm}
\usepackage{graphicx}
\usepackage{xcolor}
\usepackage{tikz}
\usepackage{booktabs}
\usepackage[hidelinks]{hyperref}
\biboptions{numbers,sort&compress}
\graphicspath{{./}}

\newtheorem{theorem}{Theorem}
\newtheorem{lemma}{Lemma}

\newtheorem{corollary}{Corollary}

\begin{document}

\begin{frontmatter}

\title{Computing the Arc-Deletion Distance to Orchard Networks is NP-hard\tnoteref{t1}}
\tnotetext[t1]{Supported by the National Natural Science Foundation of China (No. 12345678).}

\author[addr1]{Peng Li}
\ead{lipeng1@hjnu.edu.cn}

\author[addr2]{Zhiwei Liu}

\author[addr2]{Yangjing Long\corref{cor1}}
\ead{Yangjing@ccnu.edu.cn}
\cortext[cor1]{Corresponding author.}

\address[addr1]{School of Mathematics and Statistics, Hanjiang Normal University, Shiyan 442000, China}
\address[addr2]{School of Mathematics, Central China Normal University, Wuhan 430079, China}

\begin{abstract}
Phylogenetic networks generalize phylogenetic trees by allowing reticulate evolutionary events such as horizontal gene transfer and hybridization. Among the many subclasses of phylogenetic networks, orchard networks have attracted increasing attention due to their structural and algorithmic properties. In this paper, we study the arc-deletion distance to orchard networks, defined as the minimum number of reticulate arcs whose deletion transforms a phylogenetic network into an orchard network. We prove that computing this distance is NP-hard via a polynomial-time reduction from the Degree-3 Vertex Cover problem. Our result establishes the computational intractability of this proximity measure and contributes to the complexity theory of phylogenetic network transformations.
\end{abstract}


\end{frontmatter}

\section{Introduction}\label{sec:int}
Phylogenetic trees describe vertical evolutionary relationships. They show how species diverge from common ancestors. However, they fail to represent non-tree-like processes such as horizontal gene transfer (HGT) and hybridization. 
Such events are collectively referred to as reticulate events and generate conflicting evolutionary signals \cite{fontaine2015, cui2013}. 
A single tree cannot capture such complexity. To address this, researchers use phylogenetic networks. These networks allow internal nodes with in-degree greater than one, making it possible to model reticulate evolution.

Phylogenetic networks come in various types, each defined by structural properties or biological relevance. Examples include tree-child networks \cite{cardona2008}, tree-based networks \cite{francis2015}, and the more recently introduced orchard networks \cite{erdos2019}. Among these, orchard networks have received growing attention. They offer two key advantages: they can be reduced to a single-leaf vertex through a sequence of cherry-picking operations, and they can be constructed by adding horizontal arcs---representing HGT---to a tree structure \cite{vaniersel2023}. These features make orchard networks both structurally expressive and computationally tractable. As a result, they have become a central focus in recent studies of phylogenetic modeling and algorithms.

In practice, many phylogenetic networks do not meet the conditions required to be orchard networks.
This motivates the following proximity measure: for a given network, what is the minimum number of arcs that must be removed to obtain an orchard network? This question gives rise to a new proximity measure---the arc-deletion distance---which quantifies how far a network deviates from the orchard class.

\begin{center}
\fbox{\begin{minipage}{0.86\linewidth}
\textbf{Problem: $E_{\mathrm{OR}}$-DISTANCE} \\
\textbf{Input:} A network $N$ over set $X$ and a natural number $k$. \\
\textbf{Question:} Can $N$ be transformed into an orchard network by deleting at most $k$ reticulate arcs?
\end{minipage}}
\end{center}

Proximity measures based on arc deletion have received limited attention in phylogenetics. Fischer et al.\cite{fischer2023} showed that computing the arc-deletion distance for edge-based networks is NP-hard. For orchard networks, Iersel et al.\ \cite{vaniersel2023} studied a related measure based on arc additions. They proved that determining the minimum number of arcs needed to transform a general phylogenetic network into an orchard network is NP-hard. A recent bachelor's thesis~\cite{susanna2022} also compared arc-deletion-based measures with those based on leaf additions.

In this paper, we investigate the arc-deletion proximity measure for orchard networks. We formally define the arc-deletion distance to the class of orchard networks and prove that computing this distance is NP-hard.

Our work differs from the arc-addition and leaf-addition approaches in \cite{vaniersel2023}. In contrast to transforming a network by adding new structure, we study transformations obtained by deleting admissible reticulate arcs while preserving the validity of the network. Although our reduction uses part of the gadget framework introduced in \cite{vaniersel2023}, the analysis of deletable arcs, zig-zag decompositions after deletions, and HGT-consistent labelings requires different arguments.

The structure of the paper is as follows. Section 2 introduces the necessary preliminaries and reviews the definitions and structural properties of orchard and tree-based networks, which provide the foundation for our problem formulation. Section 3 formally defines the arc-deletion distance problem. Sections 4 and 5 present a proof that computing this distance is NP-hard, via a polynomial-time reduction from the DEGREE-3 VERTEX COVER problem. Section 6 concludes the paper and outlines directions for future work.

\section{Preliminaries}\label{sec:pre}
In this paper, let $X$ denote a non-empty set. For a finite directed graph \( G \), we write \( V(G) \) and \( E(G) \) for the sets of vertices and arcs of \( G \), respectively. For an arc \( uv \), the vertex \( v \) is called the \emph{head} of \( uv \), denoted by \(\text{head}(uv)\), and the vertex \( u \) is called the \emph{tail}, denoted by \(\text{tail}(uv)\).

\subsection{Phylogenetic Networks}\label{subsec:phynet}

Given a non-empty set \( X \), a \emph{rooted binary phylogenetic network} on \( X \) is a directed acyclic graph (DAG) satisfying the following conditions.
This graph contains the following four types of vertices.
\vspace{-1.5ex}
\begin{itemize}
\item It has a unique \emph{root} with indegree 0 and outdegree 1;
\vspace{-1.5ex}
\item Each \emph{tree vertex} has indegree 1 and outdegree 2;
\vspace{-1.5ex}
\item Each \emph{reticulation} has indegree 2 and outdegree 1;
\vspace{-1.5ex}
\item Each \emph{leaf} has indegree 1 and outdegree 0, and is bijectively labeled by an element of \( X \).
\vspace{-1.5ex}
\end{itemize}

For simplicity, we refer to rooted binary phylogenetic networks as \emph{networks} throughout the paper. 
Based on the types of their endpoints, arcs are classified into the following three categories:

\vspace{-1.5ex}
\begin{itemize}
\item An arc whose tail is the root is called a \emph{root arc};
\vspace{-1.5ex}
\item If \( \text{head}(a) \) is a reticulation, then \( a \) is called a \emph{reticulate arc};
\vspace{-1.5ex}
\item All other arcs are called \emph{tree arcs}.
\end{itemize}

In all networks in this paper, all arcs are oriented from top to bottom.
A network that contains no reticulations is referred to as a \emph{rooted binary phylogenetic tree}, or simply a \emph{phylogenetic tree}.

\subsection{Orchard Networks}\label{subsec:orchard}

Let \( N \) be a network. Two leaves \( x \) and \( y \) of \( N \) form a \emph{cherry} \( (x, y) \) if they are siblings. They form a \emph{reticulation cherry} \( (x, y) \) if the parent \( p_x \) of \( x \) is a reticulation, and the parent of \( p_x \), denoted \( p_{p_x} \), is also the parent of \( y \).
Suppressing a vertex of indegree 1 and outdegree 1 means deleting the vertex and replacing its incident arcs by a single arc.

\begin{figure}[htbp]
 \centering
 \includegraphics[width=0.5\linewidth]{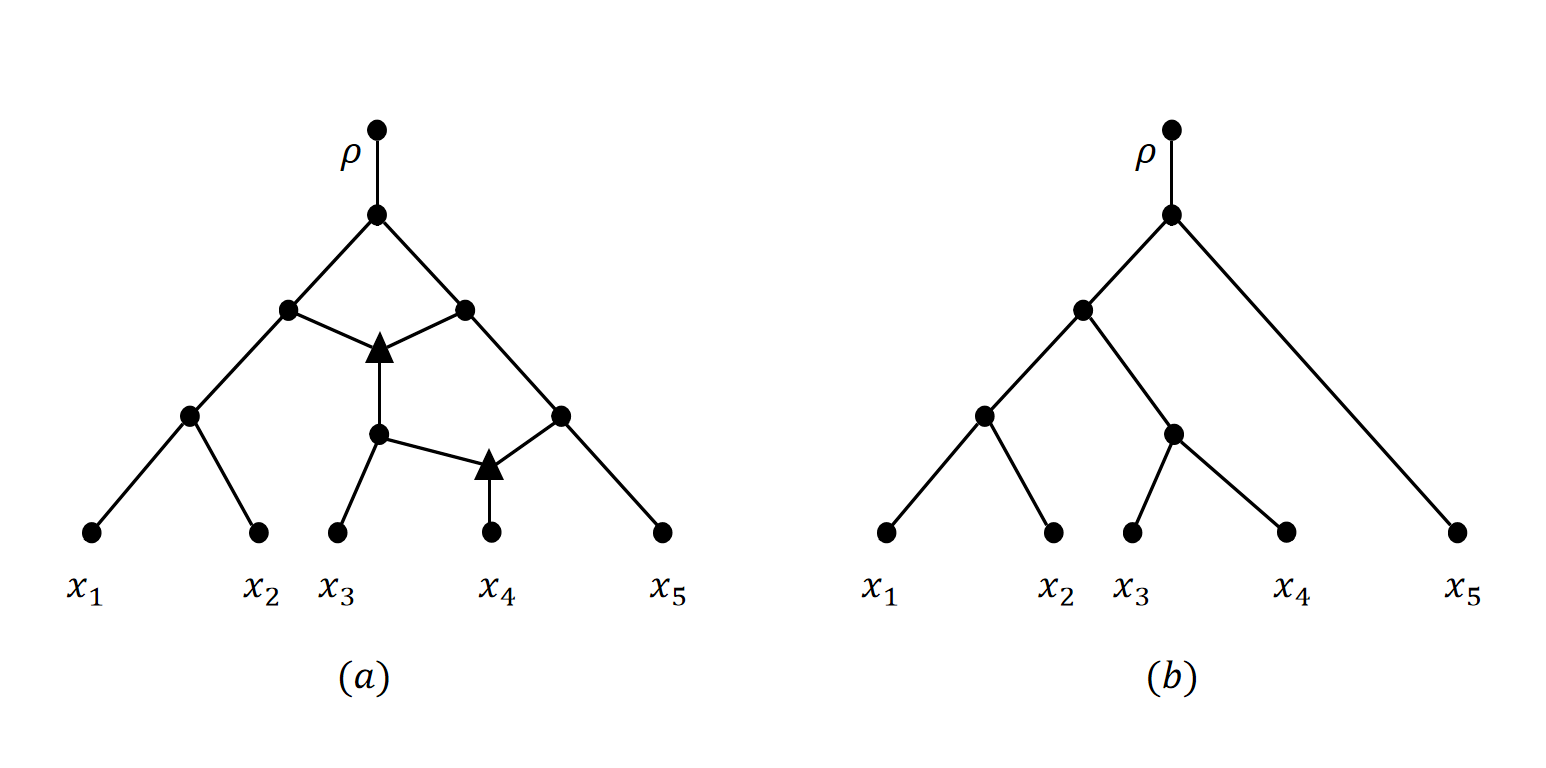}
\caption{
(a) Phylogenetic network $N_1$ with leaf set $\{x_1, x_2, \ldots, x_5\}$. $\{x_1, x_2\}$ forms a cherry, while $\{x_4, x_5\}$ forms a reticulated cherry;
(b) Phylogenetic tree $T_1$ with leaf set $\{x_1, x_2, \ldots, x_5\}$.
Round and triangle vertices denote tree vertices and reticulations, respectively. 
}

 \label{fig:reticulated}
\end{figure}

As illustrated in Fig.~\ref{fig:reticulated}, a phylogenetic network may contain reticulation vertices, and the figure provides explicit examples of a cherry and a reticulated cherry within such a network.

\noindent
We define two simplification operations on $N$ as follows:
\begin{itemize}
 \item \textbf{Reducing a cherry} $(x, y)$: 
 delete one of the leaves, say $x$, and then suppress its parent $p_x$; that is, 
 if $p_x$ now has indegree~1 and outdegree~1, remove $p_x$ and connect its parent 
 directly to its remaining child.
 \vspace{0.5ex}
 \item \textbf{Reducing a reticulation cherry} $(x, y)$: 
 delete the arc from $p_y$ to $p_x$. 
 After this deletion, suppress any vertex (typically $p_x$ and $p_y$) whose indegree 
 and outdegree are both~1.
\end{itemize}

Let \( N(x, y) \) denote the network obtained from \( N \) by applying the reduction to the pair \( (x, y) \). A network \( N \) is \emph{orchard} if there exists a finite ordered sequence \( S = \{(x_1, y_1), (x_2, y_2), \ldots\} \) such that applying cherry or reticulation cherry reductions in order reduces \( N \) to a network with a single leaf.

Let $N$ be a network and let 
$S = \{(x_1, y_1), (x_2, y_2), \ldots, (x_k, y_k)\}$ 
be an ordered sequence of pairs of leaves. 
We denote by $N_S$ 
the network obtained from $N$ by successively applying the 
cherry or reticulation cherry reductions corresponding to the pairs in $S$ 
in the given order, see Fig.~\ref{fig:N1S} for example.

\begin{figure}[htbp]
 \centering
 \includegraphics[width=0.6\linewidth]{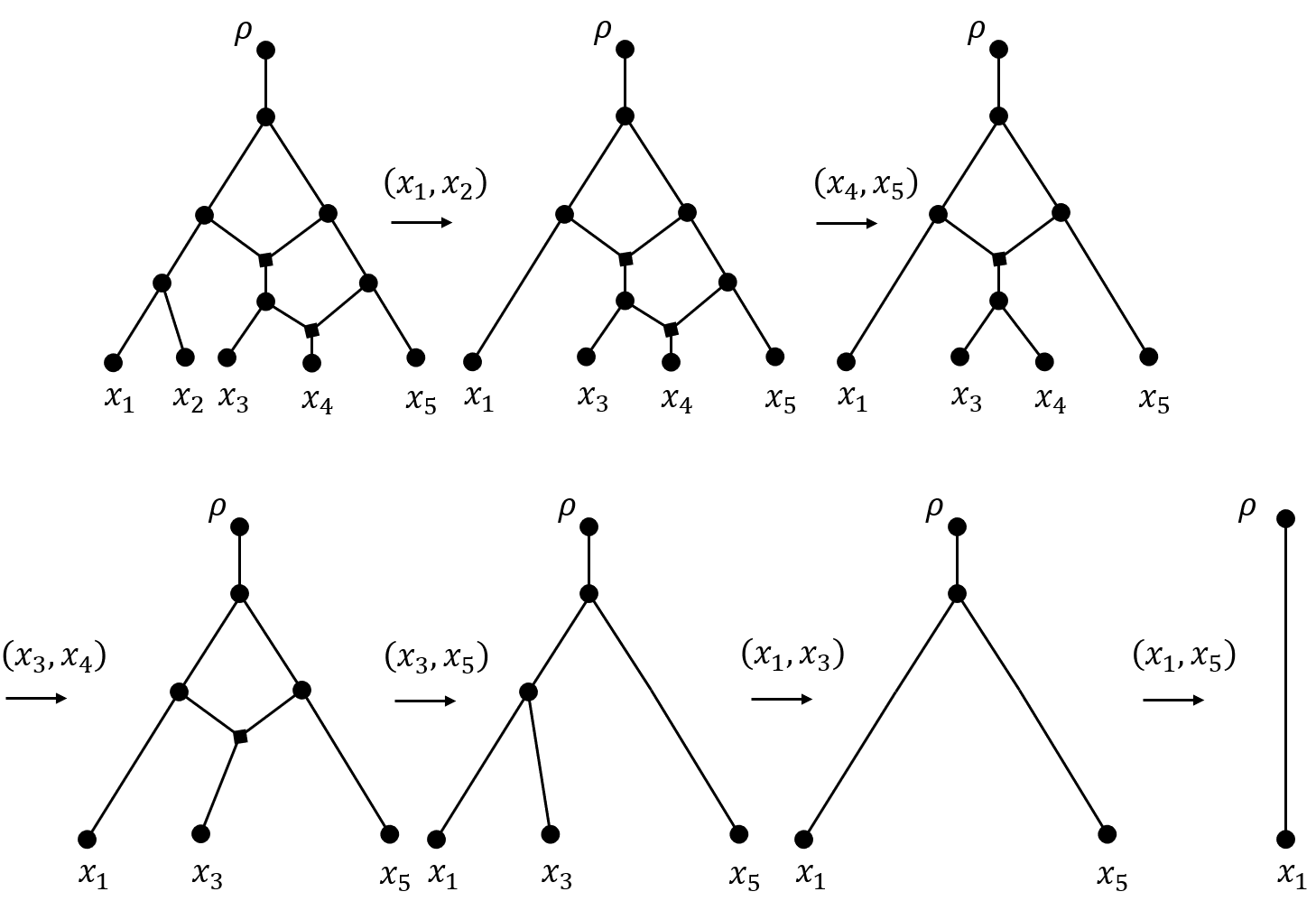} 
\caption{$(N_2)_S$ is a phylogenetic network containing only a leaf. 
$S = \{ (x_1, x_2), (x_4, x_5), (x_3, x_4), (x_3, x_5), (x_1, x_3), (x_1, x_5) \}$.}
 \label{fig:N1S}
\end{figure}

In addition to the reduction-based definition above, orchard networks can also be characterized using arc decompositions.

Let $N$ be a network. A \emph{cherry shape} is a subgraph of $N$ consisting of vertex set $\{x, y, p\}$ and arc set $\{p x, p y\}$, where $x$ and $y$ are called the \emph{endpoints}, and $p$ is the \emph{internal vertex}. We denote a cherry shape by $C$.
A \emph{reticulation cherry shape} is a subgraph of $N$ consisting of vertex set $\{x, y, p_x, p_y\}$ and arc set $\{p_x x, p_y p_x, p_y y\}$, where $x$ and $y$ are the \emph{endpoints}, $p_x$ and $p_y$ are \emph{internal vertices}, and $p_x$ is a reticulation. We denote a reticulation cherry shape by $R$.

An arc $a$ of $N$ is \emph{covered} by a cherry shape $C$ or a reticulation cherry shape $R$ if $a$ is an arc of $C$ or $R$, respectively. That is, $a \in E(C)$ or $a \in E(R)$.

A \emph{cherry cover} of $N$ is a collection $P = \{C_1, \dots, C_k, R_1, \dots, R_m\}$ of cherry shapes and reticulation cherry shapes such that every arc of $N$, except the root arc, is covered exactly once by some shape in $P$.

We further define the \emph{cherry covering auxiliary graph} $AG = (W, A)$ of $P$ as follows: each cherry or reticulation cherry shape $B \in P$ corresponds to a vertex $v_B \in W$; for any two shapes $B, B' \in P$, if an internal vertex of $B'$ is an endpoint of $B$, then there is an arc $v_Bv_{B'}$ in $A$. In this case, we say $B$ is \emph{above} $B'$.

If $AG$ contains a directed cycle, then $P$ is \emph{cyclic} ; otherwise, it is \emph{acyclic}.

The following characterization follows from \cite{vaniersel2021}.

\begin{theorem}\label{thm:acyclic}
Let $N$ be a network. Then $N$ is an orchard network if and only if it admits an acyclic cherry cover.
\end{theorem}

\begin{figure}[htbp]
 \centering
 \includegraphics[width=0.4\linewidth]{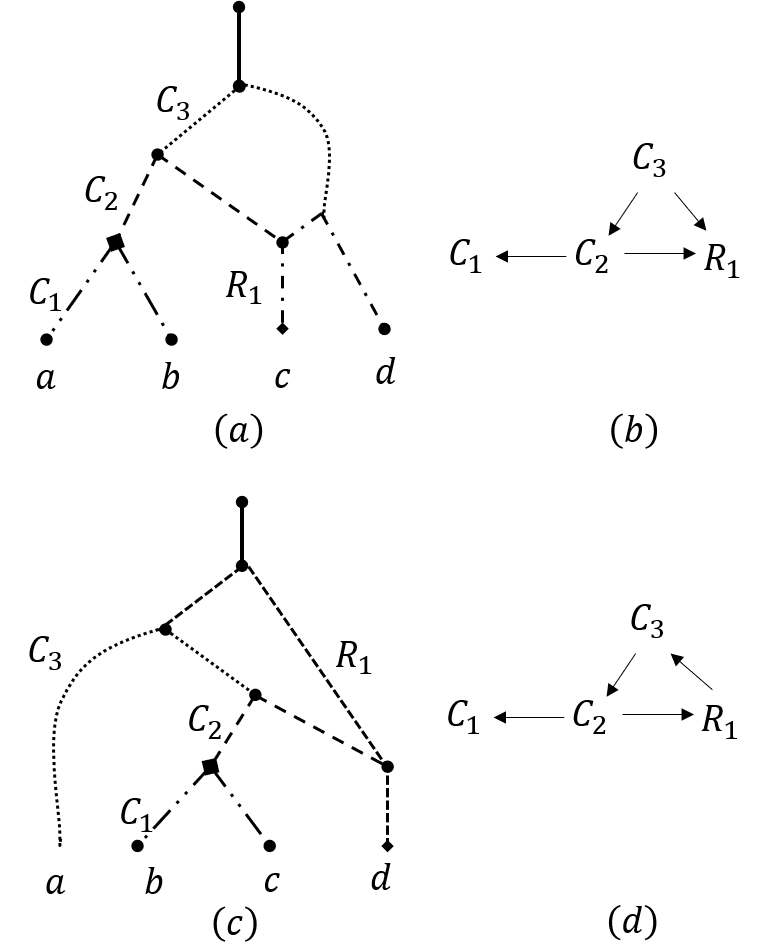}
 \caption{(a) Phylogenetic network $N_3$, with its cherry cover $P_1 = \{C_1, C_2, C_3, R_1\}$, represented by distinct dashed styles;
 (b) The cherry cover auxiliary graph of $P_1$, which is acyclic;
 (c) Phylogenetic network, with its cherry cover $P_2 = \{C_1, C_2, C_3, R_1\}$;
 (d) The cherry cover auxiliary graph of $P_2$, which is cyclic.}

 \label{fig:acyclic}
\end{figure}

Orchard networks can also be characterized by a special type of vertex labeling, called \emph{HGT-consistent labeling} (Definition 6 of \cite{vaniersel2022}).
Let $N$ be a network with vertex set $V(N)$. A function $t: V(N) \rightarrow \mathbb{R}$ is called a \emph{non-temporal labeling} of $N$ if it satisfies the following conditions:

\vspace{-1ex}
\begin{itemize}
\item For every arc $uv \in A(N)$, $t(u) \leq t(v)$ and equality holds only if $v$ is a reticulation;
\item For each internal vertex $u$, there is a child $v$ of $u$ such that $t(u) < t(v)$;
\item For every reticulation $r$ with parents $u$ and $v$, at most one of $t(u) = t(r)$ and $t(v) = t(r)$ holds.
\end{itemize}

Under a non-temporal labeling, an arc is called \emph{horizontal} if its endpoints have the same label, and \emph{vertical} otherwise.

A non-temporal labeling is said to be \emph{HGT-consistent} (horizontal gene transfer consistent) if every reticulation vertex in $N$ has exactly one incoming horizontal arc.

\begin{theorem}[Theorem 1 in \cite{vaniersel2022}] \label{thm:HGT}
Let $N$ be a network. Then $N$ is orchard if and only if it admits an HGT-consistent labeling.
\end{theorem}

\subsection{Tree-based Networks}\label{sub:TBN}

Let $N$ be a network. We say that $N$ is a \emph{tree-based network} if there exists a phylogenetic tree $T$ (called the \emph{base tree}) and a subdivision $T'$ of $T$ such that $T'$ is a spanning tree of $N$~\cite{francis2015}.

\begin{lemma}[\cite{hayamizu2021}] \label{lem:cherry}
Let $N$ be a network. Then $N$ is tree-based if and only if it has a cherry cover.
\end{lemma}

\begin{figure}[htbp]
 \centering
 \includegraphics[width=0.5\linewidth]{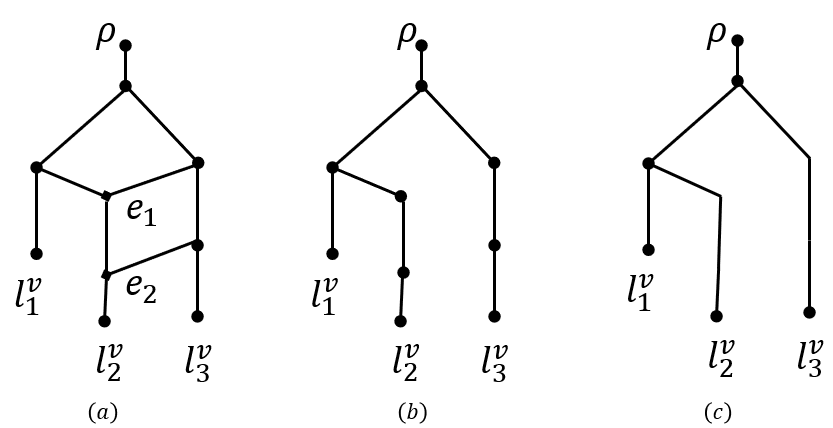}
 \caption{(a) Tree-based Phylogenetic network $N_4$; (b) A subdivision tree $\tilde{T}_4$ of $N_4$; (c) A base tree $T_4$ of $N_4$.}
 \label{fig:TBN}
\end{figure}

To characterize the structural features of tree-based networks, a method based on arc decomposition---called the \emph{maximal zig-zag trail decomposition}---was proposed in~\cite{hayamizu2021}.

Specifically, a sequence of arcs \((a_1, a_2, \ldots, a_k)\) with \(k > 1\) is called a \emph{zig-zag trail} if for every \(i \in [1, k-1]\), it holds that either \(\text{tail}(a_i) = \text{tail}(a_{i+1})\) or \(\text{head}(a_i) = \text{head}(a_{i+1})\). A zig-zag trail is \emph{maximal} if it is not properly contained in any longer zig-zag trail of the same type.
This decomposition allows us to analyze the network structure by dividing its arcs into maximal zig-zag trails.
It is known that the class of orchard networks is strictly contained within the class of tree-based networks (Corollary 1 of~\cite{vaniersel2022}). That is, every orchard network is tree-based, but not every tree-based network is orchard.
\textcolor{black}{The essential difference lies in the type of linking arcs used: orchard networks require the linking arcs to be horizontal under a non-temporal labeling, while tree-based networks allow more general linking arcs~\cite{vaniersel2022}.}

Based on the endpoints and structural properties, maximal zig-zag trails can be classified into the following four types:

\begin{itemize}
 \vspace{-1ex}
 \item \emph{Crown}: the length $k \geq 4$ is even, and either $\text{tail}(a_1) = \text{tail}(a_k)$ or $\text{head}(a_1) = \text{head}(a_k)$;
 \vspace{-1ex}
 \item \emph{M-fence}: the length $k \geq 2$ is even, not a crown, and for all $i \in [k]$, $\text{tail}(a_i)$ is a tree vertex;
 \vspace{-1ex}
 \item \emph{N-fence}: the length $k \geq 1$ is odd, with exactly one of $\text{tail}(a_1)$ or $\text{tail}(a_k)$ being a reticulation. The trail can be rearranged such that $\text{tail}(a_1)$ is a reticulation and $\text{tail}(a_k)$ is a tree vertex;
 \vspace{-1ex}
 \item \emph{W-fence}: the length $k \geq 2$ is even, with both $\text{tail}(a_1)$ and $\text{tail}(a_k)$ being reticulations.
\end{itemize}

A set $S$ of maximal zig-zag trails is called a \emph{zig-zag decomposition} of network $N$ if the trails in $S$ partition all arcs of $N$ except the root arc.

\begin{lemma}[Theorem 4.2 in \cite{hayamizu2021}] \label{lem:zig-zag}
Let $N$ be a network. Then $N$ admits a unique zig-zag decomposition.
\end{lemma}

\begin{lemma}[Corollary 4.6 in \cite{hayamizu2021}] \label{lem:Wfence}
Let $N$ be a network. Then $N$ is tree-based if and only if its zig-zag decomposition contains no W-fence.
\end{lemma}

\section{Arc-Deletion Distance to Orchard Networks}

To investigate the transformation of a general network into an orchard network, we first define a valid arc deletion operation, and based on it, introduce a new proximity measure called the \emph{arc-deletion distance}.

In a network, arc deletions must preserve the network's validity. Therefore, not all arcs are eligible for deletion. We identify the following two types of illegal arc deletions:

\begin{itemize}
 \item \textbf{Incoming arcs to tree vertices cannot be deleted:} such deletions may disconnect the network or produce a vertex with indegree 0 and outdegree 2, which violates the network definition;
 \item \textbf{Outgoing arcs from reticulation vertices cannot be deleted:} deleting such arcs may create structures with indegree 2 and outdegree 0 or break the connectivity of the graph, which are also invalid.
\end{itemize}

An arc uv is deletable if:

\begin{enumerate}
 \item the tail vertex $u$ of the arc satisfies an outdegree equal to 2;
 \item the head vertex $v$ is a reticulation.
\end{enumerate}

After deleting an arc, we apply the following operations:
\begin{itemize}
 \item Suppress all vertices with both indegree and outdegree equal to 1;
 \item Remove all parallel arcs, keeping only one of them.
\end{itemize}

Based on the above rules, we define the proximity measure to orchard networks via arc deletions for a given network $N = (V, A)$ as:

\[
E_{\mathrm{OR}}(N) = \min \left\{ |A'| \,\middle|\, A' \subseteq A,\; N - A' \text{ is an orchard network} \right\}
\]

where $N - A'$ denotes the network obtained by deleting the arc set $A'$ from $N$.

\begin{theorem}\label{main theorem}
Let $N$ be a network. Computing $E_{\mathrm{OR}}(N)$ is NP-hard.
\end{theorem}

The remainder of this paper is devoted to the proof of Theorem \ref{main theorem}.

\section{The Reduction}\label{sec:reduction}

In this section, we present a polynomial-time reduction from the \textsc{Degree-3 Vertex Cover} problem to the \textsc{$E_{\mathrm{OR}}$-Distance} problem.

\vspace{8pt}
\fbox{
\parbox{0.9\linewidth}{
\textbf{DEGREE-3 VERTEX COVER (DECISION)} \\
\textbf{Input:} A 3-regular graph $G = (V, E)$ and a natural number $k$. \\
\textbf{Question:} Does there exist a vertex cover $V' \subseteq V$ such that $|V'| \leq k$?
}
}
\vspace{10pt}

\fbox{
\parbox{0.9\linewidth}{
\textbf{\textsc{$E_{\mathrm{OR}}$-DISTANCE (DECISION)}} \\
\textbf{Input:} A phylogenetic network $N$ and a natural number $k$. \\
\textbf{Question:} Does there exist a subset of arcs $A' \subseteq A(N)$ with $|A'| \leq k$ such that $N - A'$ is an orchard network?
}
}
\vspace{10pt}

\begin{figure}[htbp]
 \centering
 \includegraphics[width=0.65\linewidth]{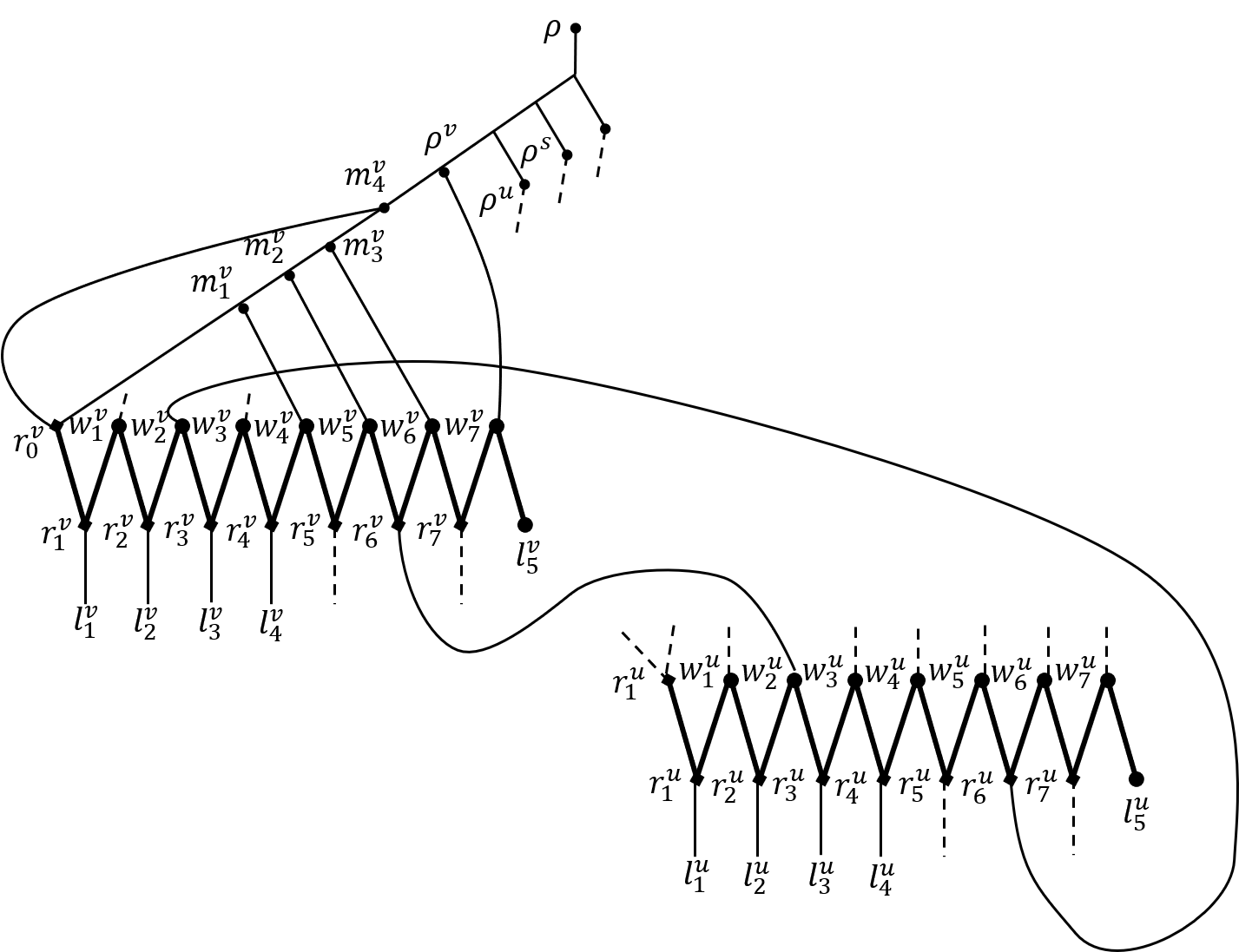}
 \caption{The two gadgets corresponding to an edge $uv\in E(G)$ and the inter-gadget arcs between them, following the construction of \cite{vaniersel2023}.}
 \label{fig:phy_net}
\end{figure}

We reduce from the \textsc{Degree-3 Vertex Cover} problem.
Let $(G,k)$ be an instance, where $G=(V(G),E(G))$ is a 3-regular graph.
We construct a phylogenetic network $N_G$ such that deleting at most $k$
reticulate arcs makes $N_G$ orchard if and only if $G$ has a vertex cover of size at most $k$.
For each vertex $v\in V(G)$, we construct a gadget $\operatorname{Gad}(v)$.

\medskip
\noindent
\textbf{The principal part.}

The principal part of $\operatorname{Gad}(v)$ is an $N$-fence of length $15$.
Note that 
it consists of the arcs
$\{r^v_0 r^v_1, w^v_i r^v_i, w^v_i r^v_{i+1} (1\le i\le 6),
 w^v_7 r^v_7, w^v_7 l^v_5\}$.
For each $i\in[4]$, the reticulation $r^v_i$ has leaf child $l^v_i$.

\medskip
\noindent
\textbf{The attachment structure.}

Above the principal part, we add tree vertices
$m^v_1,m^v_2,m^v_3,m^v_4,\rho^v$.
The arcs are
$m^v_1 r^v_0, m^v_1 w^v_4, $ 
$m^v_2 m^v_1, m^v_2 w^v_5,$ 
$m^v_3 m^v_2, m^v_3 w^v_6,$
$m^v_4 m^v_3, m^v_4 r^v_0$
and $\rho^v m^v_4, \rho^v w^v_7.$
We call the union of the principal part with the arcs
$m^v_1r^v_0$ and $m^v_4r^v_0$ the extended principal part of
$\operatorname{Gad}(v)$.

\medskip
\noindent
\textbf{Connecting the gadgets.}

Let $V(G)=\{v_1,v_2,\ldots,v_g\}$. We add a global root $\rho$ and tree vertices $s_1,s_2,\ldots,s_{g-1}$. Add arcs
$\rho s_1, s_i s_{i+1}(1\le i\le g-2), s_{g-1}\rho^{v_g}, $
and, for each $i\in[g-1]$, add the arc
$s_i\rho^{v_i}.$
Thus all gadgets are placed below the common root $\rho$.

\medskip
\noindent
\textbf{Inter-gadget arcs.}

For every vertex $v\in V(G)$, let $N_G^{\mathrm{graph}}(v)$ denote the set of neighbours of $v$ in the graph $G$. Fix bijections
\[
\pi_v:N_G^{\mathrm{graph}}(v)\rightarrow \{1,2,3\}
\quad\text{and}\quad
\tau_v:N_G^{\mathrm{graph}}(v)\rightarrow \{5,6,7\}.
\]
For each edge $uv\in E(G)$, add the two arcs
$r^u_{\tau_u(v)}w^v_{\pi_v(u)}$
and
$r^v_{\tau_v(u)}w^u_{\pi_u(v)}.$
These arcs are outgoing from reticulations and enter tree vertices. Hence,
by the definition of deletable arcs, they are not deletable.
Finally, the leaves of $N_G$ are
$\{l^v_i:v\in V(G),i\in[5]\}.$
This completes the construction of $N_G$.

\section{Proof of Theorem~\ref{main theorem}}\label{sec:proof-main}

The goal of this section is to prove Theorem~\ref{main theorem}.
To this end, we establish a sequence of structural lemmas about the network $N_G$.

Let $N_G$ be the directed graph constructed by the above reduction. It is easy to verify that $N_G$ is a directed acyclic graph that satisfies the definition of a phylogenetic network. Its unique root is $\rho$, and its leaf set is $\{l^{v}_i : i \in [5], v \in V(G)\}$. Furthermore, there exists a maximal zig-zag path decomposition of $A(N_G)$ such that all arcs in $N_G$ can be partitioned into $M$-fences and $N$-fences. Using the structural properties of zig-zag decompositions (Lemma~\ref{lem:Wfence}), we obtain the following lemma:

\begin{lemma}[Observation 12 of \cite{vaniersel2023}] \label{lem:3regular}
Let $G$ be a 3-regular graph, and let $N_G$ be the network constructed by the above reduction. Then $N_G$ is tree-based.
\end{lemma}

By Lemma~\ref{lem:3regular}, we confirm that $N_G$ is a tree-based network. Next, we investigate the structural properties of $N_G$, particularly those related to cherry-picking coverings, which will play a crucial role in the subsequent complexity analysis. According to Lemma~\ref{lem:cherry}, since $N_G$ is tree-based, it must admit a cherry-picking covering.

To proceed, we introduce some auxiliary notation and lemmas.
Let $N$ be a network, and let
$\widehat{N}_i=(a^{i}_1,a^{i}_2,\ldots,a^{i}_{k_i})$
For each$j\in \left[1,\frac{k_i-1}{2}\right],$
let $c^{i}_{2j-1}$ denote the child vertex of the arc
$a^{i}_{2j+1}$.
Equivalently,$c^{i}_{2j-1}=\mathrm{head}(a^{i}_{2j+1}).$

Moreover, we use the extended principal part of $\operatorname{Gad}(v)$ as defined in Section~\ref{sec:reduction}, namely the principal $N$-fence together with the arcs $m_1^v r_0^v$ and $m_4^v r_0^v$.

The following lemma describes the relationship between $N$-fences and cherry covers.

\begin{lemma} 
[Lemma 13 of \cite{vaniersel2023}] \label{lem:N-fence}
Let $N$ be a tree-based network, and let
$\widehat{N}_1,\widehat{N}_2,\ldots,\widehat{N}_n$ be the $N$-fences in $N$, each of length at least $3$.
Then every cherry cover of $N$ must contain the reticulated cherry shapes of the form
\[
\{\text{head}(a_{2j-1}^{i})c_{2j-1}^{i},\ a_{2j}^{i},\ a_{2j+1}^{i}\},
\]
for each $i\in[n]$ and $j\in\left[1,(k_i-1)/2\right]$, where $k_i$ is the length of $\widehat{N}_i$.
\end{lemma}

This lemma implies that although $N$ may admit multiple different cherry covers, the shapes of the reticulated cherries used to cover each $N$-fence in $N$ are uniquely determined. Specifically, for every vertex $v \in V(G)$, the principal part of $\text{Gad}(v)$ is an $N$-fence of length 15, denoted as $N = (a_1, a_2, \ldots, a_{15})$. In every cherry cover of $N$, the reticulated cherries covering this $N$-fence follow a unique pattern, denoted as $\mathcal{R}v$, consisting of cherries of the form \{$\text{head}(a_{i-1})$ $c_{i-1}$, $a_i$, $a_{i+1}$\} for even $i \in [2, 14]$.

We now prove Lemma 4, which establishes the relationship between $N$-fences and cherry cover auxiliary graphs.

\begin{lemma}
[Lemma 14 of \cite{vaniersel2023}] \label{lem:path}
Let $N$ be a tree-based network, and suppose that $N$ contains two $N$-fences $\widehat{N}_u := (a_1^u, a_2^u, \ldots, a_{k_{u}}^u)$ and $\widehat{N}_v := (a_1^v, a_2^v, \ldots, a^{v}_{k_{v}})$, each of length at least 3. If there exists a directed path in $N$ from $\text{head}(a_h^u)$ to $\text{tail}(a_i^v)$ and another directed path from $\text{head}(a_j^v)$ to $\text{tail}(a_k^u)$, where $h$, $i$, $j$, and $k$ are even integers satisfying $k < h$ and $i < j$, then every cherry cover auxiliary graph of $N$ contains a cycle.
\end{lemma}

This lemma highlights that although $N_G$ may admit multiple different cherry covers, the reticulated cherry shapes used to cover the $N$-fences induce fixed structural relationships in the cherry cover auxiliary graph. In particular, for every edge $uv \in E(G)$, the sets of reticulated cherries that cover the principal parts of $\text{Gad}(u)$ and $\text{Gad}(v)$ collectively form a cycle in the auxiliary graph corresponding to the cherry cover.

\begin{lemma}\label{lem:Gad}
Let $G$ be a 3-regular graph, and let $N_{G}$ be the network constructed above.
Let $A$ be a set of deletable reticulate arcs such that $N_G-A$ is orchard. Then, for every edge $uv\in E(G)$, the set $A$ contains an arc in the extended principal part of $\operatorname{Gad}(u)$ or in the extended principal part of $\operatorname{Gad}(v)$.
\end{lemma}

\begin{proof}
Suppose not. Then there exists an edge $uv\in E(G)$ such that no arc of $A$ lies in the extended principal part of either $\operatorname{Gad}(u)$ or $\operatorname{Gad}(v)$.
Let $\widehat N_u=(a^u_1,a^u_2,\ldots,a^u_{15})
\quad\text{and}\quad \widehat N_v=(a^v_1,a^v_2,\ldots,a^v_{15})$
be the principal $N$-fences in $\operatorname{Gad}(u)$ and
$\operatorname{Gad}(v)$, respectively. Since no arc in the corresponding
extended principal parts is deleted, both $N$-fences remain present in $N_G-A$.
By construction, the edge $uv$ gives two directed inter-gadget arcs:
one from a reticulation in $\operatorname{Gad}(u)$ to a tree vertex in
$\operatorname{Gad}(v)$, and one from a reticulation in
$\operatorname{Gad}(v)$ to a tree vertex in $\operatorname{Gad}(u)$.
These arcs are outgoing from reticulation vertices and hence are not deletable by the definition of deletable arcs. Therefore both directed
connections remain in $N_G-A$.
Consequently, in $N_G-A$ there is a directed path from a later part of $\widehat N_u$ to an earlier part of $\widehat N_v$, and also a directed
path from a later part of $\widehat N_v$ to an earlier part of
$\widehat N_u$. More precisely, these paths satisfy the hypotheses of
Lemma~\ref{lem:path}: there exist even indices $h,i,j,k$ with $k<h$ and $i<j$ such that
there is a directed path from $\operatorname{head}(a^u_h)$ to
$\operatorname{tail}(a^v_i)$ and a directed path from
$\operatorname{head}(a^v_j)$ to $\operatorname{tail}(a^u_k)$. 
By Lemma \ref{lem:path},
every cherry-cover auxiliary graph of $N_G-A$ contains a
directed cycle. Hence $N_G-A$ admits no acyclic cherry cover. By
Theorem \ref{thm:acyclic}, $N_G-A$ is not orchard, contradicting the assumption.
Therefore, for every edge $uv\in E(G)$, at least one arc must be deleted
from the extended principal part of $\operatorname{Gad}(u)$ or
$\operatorname{Gad}(v)$.

\end{proof}

Next, we will prove that if $G$ has a vertex cover $V_{\text{sol}}$ of size $k$, then $N_G$ can be made orchard by deleting $k$ reticulate arcs.
For every $v \in V_{\text{sol}}$, we delete the arc $w_{4}^{v}r_{4}^{v}$ in Gad($v$), and denote the resulting network as $N_{G}^{E}$.

\begin{figure}[htbp]
 \centering
 \includegraphics[width=0.8\linewidth]{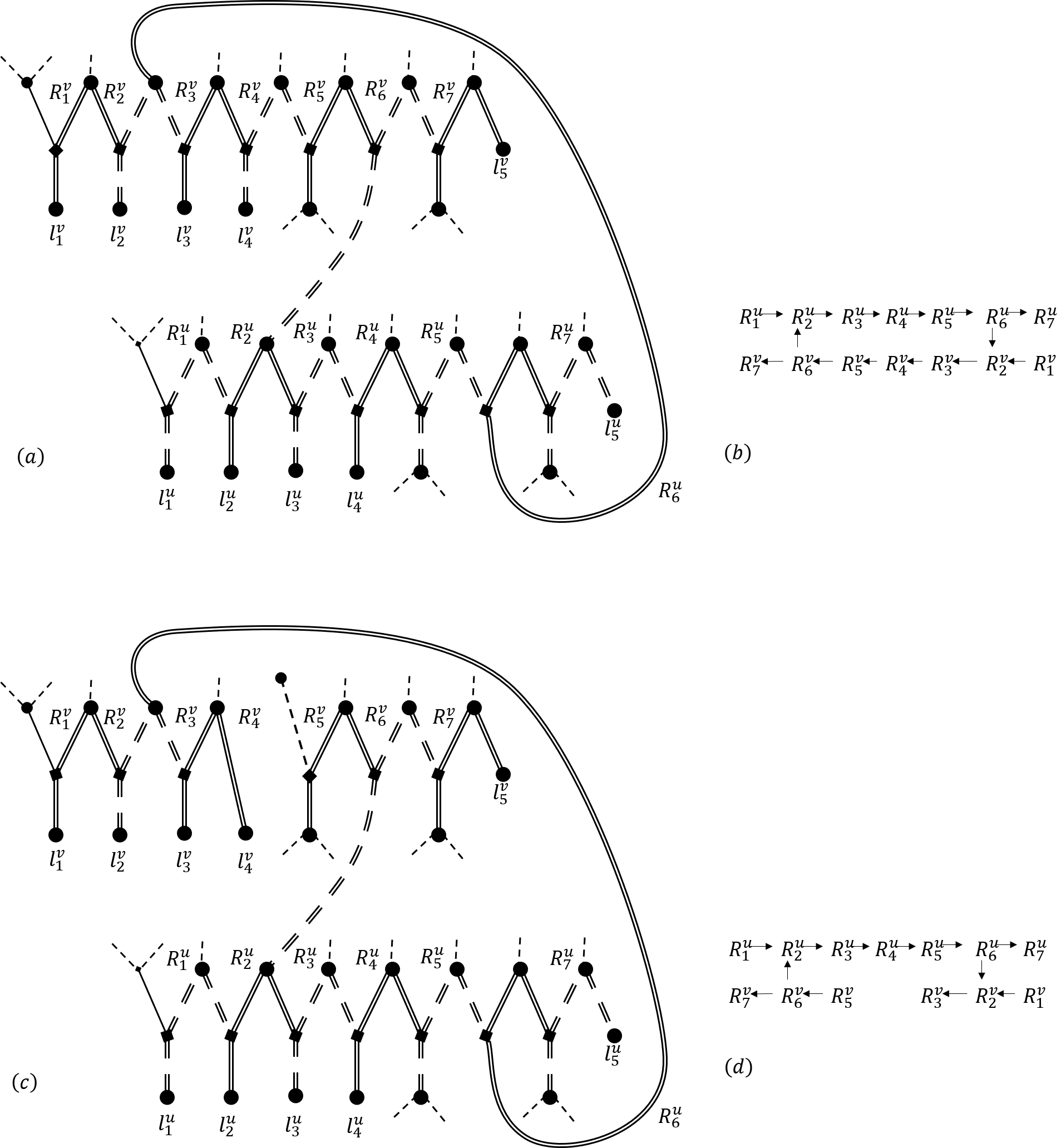}
 \caption{The effect of deleting the arc $w^v_4r^v_4$ in a selected gadget. Panels (a) and (c) show the relevant parts of $N_G$ and $N_G^E$, respectively. Panels (b) and (d) show the corresponding parts of the cherry-cover auxiliary graphs; the directed cycle present in (b) is destroyed in (d).}

 \label{fig:example}
\end{figure}

\begin{lemma}\label{lem:N_G}
Let $V_{\mathrm{sol}}$ be a vertex cover of $G$. Let $N^E_G$ be the network
obtained from $N_G$ by deleting, for each $v\in V_{\mathrm{sol}}$, the arc
$w^v_4 r^v_4$. Then $N^E_G$ is tree-based.
\end{lemma} 

\begin{proof}
We use the characterization of tree-based networks by maximal zig-zag
decompositions. By construction, $N_G$ has a zig-zag decomposition whose
components are $M$-fences and $N$-fences, and hence it contains no
$W$-fence. Consider a vertex $v\in V_{\mathrm{sol}}$. In the principal $N$-fence of
$\operatorname{Gad}(v)$, deleting the reticulate arc $w^v_4r^v_4$ breaks
this $N$-fence into two maximal zig-zag trails. One is an $M$-fence and the
other is an 
$N$-fence. No $W$-fence is created, because neither of the two resulting trails has both end arcs incident with reticulations in the
configuration required for a $W$-fence.
All other maximal zig-zag trails are unchanged, except possibly for being
truncated at the deleted arc. Such truncation cannot create a new $W$-fence; it only shortens an existing $M$- or $N$-fence. Therefore the zig-zag decomposition of $N^E_G$ contains no $W$-fence. By Lemma~\ref{lem:Wfence},
$N^E_G$ is tree-based.
\end{proof}

\begin{lemma} \label{lem:N_GO}
Let $V_{\mathrm{sol}}$ be a vertex cover of $G$, and let $N^E_G$ be obtained
from $N_G$ by deleting $w^v_4r^v_4$ for every $v\in V_{\mathrm{sol}}$.
Then $N^E_G$ is orchard.
\end{lemma}

\begin{proof}
By Theorem \ref{thm:HGT}, it suffices to construct an HGT-consistent labeling of
$N^E_G$.
We define a labeling $t:V(N^E_G)\to \mathbb R$ as follows. First set
$t(\rho)=0$. Along the tree structure connecting the gadgets, assign labels
strictly increasing from each parent to each child. In particular, the
vertices $s_i$, $\rho^v$,$m^v_4$,$m^v_3$, and $m^v_2$ are assigned labels strictly larger than the label of their respective parents. After these
labels are assigned, subtract a sufficiently large constant from all of
them so that they are all strictly negative.
For every $v\in V(G)$, set
$t(m^v_1)=t(r^v_0)=0.$

The remaining vertices in the principal part are labeled as follows. 
If $v\in V_{\mathrm{sol}}$, define
$t(r^v_1)=t(w^v_1)=12,
t(r^v_2)=t(w^v_2)=13,
t(r^v_3)=t(w^v_3)=14,$
and $t(r^v_5)=t(w^v_5)=2,
t(r^v_6)=t(w^v_6)=3,
t(r^v_7)=t(w^v_7)=4.$
Define $t(r^v_4)=14$. Since the arc $w^v_4 r^v_4$ has been deleted, the reticulation $r^v_4$ has the unique horizontal incoming arc $w^v_3 r^v_4$.
 If $v\notin V_{\mathrm{sol}}$, define
$t(r^v_i)=t(w^v_i)=i+4\quad\text{for every }i\in[7].$

Finally, for every leaf $l$ with parent $p_l$, set
$t(l)=t(p_l)+1.$

We verify that $t$ is HGT-consistent. First, every arc is non-decreasing
with respect to $t$. Tree arcs are assigned strictly increasing labels.
For each reticulation $r$, one incoming arc is horizontal, and every other incoming arc is vertical. Thus equality occurs only on arcs entering reticulations.
Second, every internal vertex has at least one child with strictly larger
label. This is immediate for the tree vertices in the global connecting
structure. For the vertices inside each gadget, the assignment above
ensures that each tree vertex has either a vertical child in the principal
part or a leaf descendant with larger label.
Third, each reticulation has exactly one incoming horizontal arc. For
reticulations inside a gadget not selected by the vertex cover, this
horizontal arc is the prescribed arc $w^v_i r^v_i$. For gadgets selected by the vertex cover, the deletion of $w^v_4r^v_4$ leaves $w^v_3r^v_4$ as the unique horizontal incoming arc to $r^v_4$. For $r^v_0$, the unique
horizontal incoming arc is $m^v_1r^v_0$.
It remains to check the inter-gadget arcs. Let $uv\in E(G)$. Since $V_{\mathrm{sol}}$ is a vertex cover, at least one of $u$ and $v$ belongs to $V_{\mathrm{sol}}$.
Consider an inter-gadget arc entering $\operatorname{Gad}(v)$, say
\[
r^u_{\tau_u(v)}w^v_{\pi_v(u)}.
\]
If $v\in V_{\mathrm{sol}}$, then $w^v_{\pi_v(u)}\in\{w^v_1,w^v_2,w^v_3\}$ has label $12$, $13$, or $14$. The tail $r^u_{\tau_u(v)}$ belongs to $\{r^u_5,r^u_6,r^u_7\}$, whose label is either one of $2,3,4$ if $u\in V_{\mathrm{sol}}$, or one of $9,10,11$ if $u\notin V_{\mathrm{sol}}$. Hence this inter-gadget arc is strictly increasing.
If $v\notin V_{\mathrm{sol}}$, then $u\in V_{\mathrm{sol}}$. In this case the tail lies in $\{r^u_5,r^u_6,r^u_7\}$ and has label $2$, $3$, or $4$, while the head lies in $\{w^v_1,w^v_2,w^v_3\}$ and has label $5$, $6$, or $7$. Hence the arc is again strictly increasing. Therefore all inter-gadget arcs are vertical and cannot create an additional horizontal incoming arc to any reticulation. 
Therefore $t$ is a non-temporal labeling in which every reticulation has
exactly one incoming horizontal arc. Thus $t$ is HGT-consistent. By
Theorem \ref{thm:HGT}, $N^E_G$ is orchard.
\end{proof}

\begin{lemma} \label{lem:min}
Let $G$ be a 3-regular graph, and let $N_G$ be the network obtained by the above reduction. Then, for every integer $k$, the graph $G$ has a vertex cover of size at most $k$ if and only if
\[
E_{\mathrm{OR}}(N_G)\le k.
\]
\end{lemma}

\begin{proof}
Suppose first that $G$ has a vertex cover $V_{\mathrm{sol}}$ with $|V_{\mathrm{sol}}|\le k$. For every $v\in V_{\mathrm{sol}}$, delete the arc $w^v_4r^v_4$ in $\operatorname{Gad}(v)$. By Lemma~\ref{lem:N_GO}, the resulting network is orchard. Hence
\[
E_{\mathrm{OR}}(N_G)\le |V_{\mathrm{sol}}|\le k.
\]

Conversely, suppose that $E_{\mathrm{OR}}(N_G)\le k$. Then there exists a set $A$ of at most $k$ deletable arcs such that $N_G-A$ is orchard. Let $P_v$ denote the set of arcs in the extended principal part of $\operatorname{Gad}(v)$, and define
\[
V_A=\left\{v\in V(G): A\cap P_v\ne\emptyset\right\}.
\]
By Lemma~\ref{lem:Gad}, for every edge $uv\in E(G)$, the set $A$ contains an arc in the extended principal part of $\operatorname{Gad}(u)$ or in the extended principal part of $\operatorname{Gad}(v)$. Therefore at least one of $u$ and $v$ belongs to $V_A$. Hence $V_A$ is a vertex cover of $G$.

Moreover, each vertex in $V_A$ is charged to at least one deleted arc, and so
\[
|V_A|\le |A|\le k.
\]
Thus $G$ has a vertex cover of size at most $k$.
\end{proof}

\medskip
\par \noindent \textbf{Proof of Theorem \ref{main theorem}: }

Let $A_{\text{sol}} \subseteq A(N)$ be a set of at most $k$ arcs. Define $N' = N - A_{\text{sol}}$. Whether $N'$ is an orchard network can be verified in polynomial time \cite{janssen2021}, so the $E_{\mathrm{OR}}$-Distance Decision Problem belongs to NP.

Furthermore, consider an instance $(G, k)$ of the vertex cover decision problem on 3-regular graphs. As described in the reduction above, we can construct an instance $(N_G, k)$ of the $E_{\mathrm{OR}}$-Distance Decision Problem by applying a polynomial-time transformation to $G$. The construction has size polynomial in $|V(G)|+|E(G)|$, since each gadget contains a constant number of vertices and arcs, and each edge of $G$ contributes exactly two inter-gadget arcs.

By Lemma \ref{lem:min}, this reduction preserves the answer: $(G, k)$ is a yes-instance of the 3-regular vertex cover problem if and only if $(N_G, k)$ is a yes-instance of the $E_{\mathrm{OR}}$-Distance Decision Problem.

Since the 3-regular vertex cover decision problem is NP-complete, and the reduction is polynomial, it follows that the $E_{\mathrm{OR}}$-Distance Decision Problem is NP-hard. Combining this with the membership in NP, the problem is NP-complete.

As a consequence, the corresponding optimization problem of computing $E_{\mathrm{OR}}(N)$---that is, the minimum number of reticulate arcs that must be removed from $N$ to make it orchard---is NP-hard.

\begin{corollary}\label{cor:decision}
The decision problem $E_{\mathrm{OR}}$-Distance is NP-complete.
\end{corollary}

\section{Discussion}\label{sec:discussion}

In this paper, we established that computing the minimum number of reticulate arc deletions needed to transform a network into an orchard network is NP-hard. The proof is obtained by a polynomial-time reduction from the vertex cover problem on 3-regular graphs. This result confirms that the arc-deletion-based proximity measure to orchard networks is computationally intractable.

Several open problems remain for future investigation. One direction is to compare different proximity measures for orchard networks. In particular, it would be interesting to understand whether there is a quantitative relationship between the arc-deletion distance and the leaf-addition distance. Such a comparison may depend strongly on structural restrictions, for example whether the network is binary, tree-based, or belongs to another well-behaved subclass of phylogenetic networks.

Another direction concerns the comparison between arc deletion and vertex deletion.
Given a rooted phylogenetic network $N$, one may define the vertex-deletion distance to orchard networks as the minimum number of vertices whose removal, followed by the necessary suppressions and simplifications, transforms $N$ into an orchard network. This raises the following general question: how does this vertex-deletion distance compare with the arc-deletion distance studied in this paper?

There are two possible types of relationships one may seek. First, one may ask whether a small arc-deletion distance implies a small vertex-deletion distance. This is not immediate, since deleting one vertex may remove several incident arcs, while deleting a single reticulate arc may only destroy one local obstruction. Second, one may ask whether a small vertex-deletion distance implies a small arc-deletion distance. This direction is also non-trivial, since removing a vertex can simultaneously eliminate several cycles or several non-orchard structures that may require many separate arc deletions. Therefore, arc deletion and vertex deletion may measure different kinds of obstructions to orchardness.

From an algorithmic point of view, this comparison suggests several concrete problems. One may study whether the vertex-deletion version is NP-hard, whether it admits fixed-parameter algorithms with respect to the deletion budget, and whether approximation-preserving reductions exist between the arc-deletion and vertex-deletion variants. If such reductions can be established, known approximation algorithms or parameterized techniques for classical problems such as Vertex Cover or Feedback Arc Set may become applicable to orchard-network proximity measures.

More broadly, it would be useful to investigate approximation algorithms and parameterized complexity for $E_{\mathrm{OR}}$-Distance. In particular, determining whether the problem is fixed-parameter tractable with respect to the deletion budget $k$, or whether it is APX-hard, would deepen our understanding of the computational landscape of orchard network transformations.

\section*{Declaration of generative AI and AI-assisted technologies in the writing process}
During the preparation of this work, the authors used OpenAI's ChatGPT to improve the readability and language of the manuscript and to assist with exposition and formatting. After using this tool, the authors reviewed and edited the content as needed and take full responsibility for the content of the publication.

\end{document}